\documentclass[11pt,reqno]{amsart}
\newcommand\version{November 21, 2007}

\usepackage{amsmath,amsfonts,amsthm,amssymb,amsxtra}

\newtheorem{theorem}{Theorem}[section]
\newtheorem{proposition}[theorem]{Proposition}
\newtheorem{lemma}[theorem]{Lemma}
\newtheorem{corollary}[theorem]{Corollary}

\theoremstyle{definition}

\theoremstyle{remark}

\newtheorem{remark}[theorem]{Remark}

\numberwithin{equation}{section}

\newcommand{\C}{\mathbb{C}}
\newcommand{\cl}{\mathrm{cl}}

\renewcommand{\epsilon}{\varepsilon}

\newcommand{\loc}{{\rm loc}}

\renewcommand{\phi}{\varphi}

\newcommand{\R}{\mathbb{R}}

\DeclareMathOperator{\dom}{dom}
\DeclareMathOperator{\im}{Im}
\DeclareMathOperator{\re}{Re}

\DeclareMathOperator{\tr}{tr}

\title[Schr\"odinger operators with surface potentials -- \version]{Spectral inequalities for Schr\"odinger operators with surface potentials}

\author{Rupert L. Frank}
\address{Rupert L. Frank, Department of Mathematics,
  Princeton University, Fine Hall, Princeton, NJ 08544, USA}
\email{rlfrank@math.princeton.edu}
\author{Ari Laptev}
\address{Ari Laptev, Department of Mathematics, Imperial College London, London SW7 2AZ, UK $\&$ Department of Mathematics, Royal Institute of Technology, 100 44 Stockholm, Sweden}
\email{a.laptev@imperial.ac.uk $\&$ laptev@math.kth.se}

\begin{document}

\begin{abstract}
We prove sharp Lieb-Thirring inequalities for Schr\"odinger operators with potentials supported on a hyperplane and we show how these estimates are related to Lieb-Thirring inequalities for relativistic Schr\"odinger operators.
\end{abstract}

\thanks{\copyright\, 2007 by the authors. This paper may be reproduced, in its entirety, for non-commercial purposes.}

\dedicatory{Dedicated to M. Sh. Birman on the occasion of his 80th birthday}

\maketitle


\section{Introduction}

The Cwikel-Lieb-Rozenblum and the Lieb-Thirring inequalities estimate the number and moments of eigenvalues of Schr\"odinger operators $-\Delta-V$ in $L_2(\R^N)$ in terms of an integral of the potential $V$. They state that the bound 
\begin{equation}
  \label{eq:lt}
  \tr(-\Delta-V)_-^\gamma \leq L_{\gamma,N} \int_{\R^N} V(x)^{\gamma+N/2}_+\,dx
\end{equation}
holds with a constant $L_{\gamma,N}$ independent of $V$ iff $\gamma\geq 1/2$ for $N=1$, $\gamma>0$ for $N=2$ and $\gamma\geq 0$ for $N\geq 3$. Here and below $t_\pm:=\max\{0,\pm t\}$ denotes the positive and negative part of a real number, a real-valued function or a self-adjoint operator $t$. In particular, the problem of finding the optimal value of the constant $L_{\gamma,N}$ has attracted a lot of attention recently. We refer to the review articles \cite{H2,LW2} for background information, references and applications of \eqref{eq:lt}.

The purpose of the present paper is twofold. First, we would like to find an analog of inequality \eqref{eq:lt} for Schr\"odinger operators with singular potentials $V(x)=v(x_1,\ldots,x_d)\delta(x_N)$, $d:=N-1$, supported on a hyperplane. It turns out that such an inequality is indeed valid, provided the integral on the right hand side of \eqref{eq:lt} is replaced by 
$$
	\int_{\R^d} v(x_1,\ldots,x_d)_+^{2\gamma+d}\,dx_1\ldots dx_d\, .
$$
We determine the complete range of $\gamma$'s for which the resulting inequality holds. Moreover, we find the sharp values of the constants for $\gamma\geq 3/2$ by using the method of `lifting with respect to the dimension'. This provides yet another example of the power and flexibility of this method, which was used by Laptev and Weidl \cite{LW} to obtain the sharp constants in \eqref{eq:lt} for $\gamma\geq 3/2$.

The second purpose of this paper is to point out a relation between the Schr\"odinger operator $-\Delta-v(x_1,\ldots,x_d)\delta(x_N)$ in $L_2(\R^N)$ and the relativistic Schr\"odinger operator $\sqrt{-\Delta}-v(x_1,\ldots,x_d)$ in $L_2(\R^d)$. Note that the space dimension $d=N-1$ of the relativistic operator differs from that of the non-relativistic operator. The basic idea is to relate eigenfunctions of the Schr\"odinger operator with singular potential to eigenfunctions of a non-linear eigenvalue problem involving the relativistic Schr\"odinger operator. This construction is essentially the Poisson extension and is implicit in several earlier works, e.g., in \cite{CL,FS,CS}. In our context the connection between the two operators becomes useful when combined with a monotonicity argument in the spirit of the Birman-Schwinger principle. It allows us both to prove the singular analog of inequality \eqref{eq:lt} and to (slightly) improve upon the known constants in Lieb-Thirring inequalities for relativistic Schr\"odinger operators.


\section{Schr\"odinger operators with surface potentials}\label{sec:main}

\subsection{Main results}

In this section we consider the operator
\begin{equation}\label{eq:op}
H(v)\,u=-\Delta u
\quad\text{ in } \R^{d+1}_+:=\{(x,y):\ x\in\R^d, y>0 \}
\end{equation}
together with boundary conditions of the third type
\begin{equation}\label{eq:bc}
\frac{\partial u}{\partial\nu} - v u = 0
\quad\text{on }\R^d\times\{0\}\,.
\end{equation}
Here $\partial/\partial\nu=-\partial/\partial y$ denotes the (exterior) normal derivative and $v$ a real-valued function on $\R^d$. If $v$ is form-compact with respect to $\sqrt{-\Delta}$ in $L_2(\R^d)$, then $H(v)$ can be defined as a self-adjoint operator in $L_2(\R^{d+1}_+)$ by means of the quadratic form
\begin{equation}\label{eq:form}
  \iint_{\R^{d+1}_+} |\nabla u|^2 \,dx\,dy - \int_{\R^d} v(x)|u(x,0)|^2\,dx,
  \quad u\in H^1(\R^{d+1}_+)\, .
\end{equation}
The negative spectrum of $H(v)$ consists of eigenvalues of finite multiplicities. We shall prove

\begin{theorem}[\textbf{Lieb-Thirring inequalities for surface potentials}]
  \label{main1}
  The inequality 
  \begin{equation}\label{eq:main1}
    \tr\left[ H(v) \right]_-^\gamma
    \leq S_{\gamma,d} \int_{\R^d} v(x)_+^{2\gamma+d} \,dx
  \end{equation}
  holds for all $0\leq v\in L_{2\gamma+d}(\R^d)$ iff
  \begin{equation}
    \label{eq:ltsurfass}
    \gamma>0 \quad\text{if } d=1,
    \qquad \text{and} \qquad
    \gamma\geq 0 \quad\text{if } d\geq 2 \, .
  \end{equation}
\end{theorem}

Inequality \eqref{eq:main1} reflects the correct order of growth in the strong coupling limit, as can be seen from the Weyl-type asymptotics
\begin{equation}\label{eq:weylsurf}
  \lim_{\alpha\to\infty}
  \alpha^{-2\gamma-d}\tr\left[ H(\alpha\, v) \right]_-^\gamma
  = L_{\gamma,d}^\cl \int_{\R^d} v(x)_+^{2\gamma+d}\,dx\, ,
\end{equation}
with
\begin{equation}
  \label{eq:surfclass}
  L_{\gamma,d}^\cl 
  := 2^{-d} \pi^{-d/2} \frac{\Gamma(\gamma+1)}{\Gamma(\gamma+d/2+1)}\, .
\end{equation}
Since relation \eqref{eq:weylsurf} is not completely standard we comment on its proof in Remark \ref{weyl} below.

Our second result concerns the constants in the bounds of Theorem \ref{main1}. Denoting by $S_{\gamma,d}$ the \emph{sharp} constant in \eqref{eq:main1} we infer from \eqref{eq:weylsurf} that
\begin{equation}
  S_{\gamma,d} \geq L_{\gamma,d}^\cl\, .
\end{equation}
We shall prove that for sufficiently large values of $\gamma$ one actually has equality.

\begin{theorem}[\textbf{Sharp constants}]
  \label{main2}
  Let $d\geq 1$ and $\gamma\geq 3/2$. Then the sharp constant in \eqref{eq:main1} is $S_{\gamma,d} = L_{\gamma,d}^\cl$.
\end{theorem}

We prove this theorem in Subsection \ref{sec:main2proof}. Besides the sharp constants for $\gamma\geq 3/2$ our method yields explicit and tight bounds for $S_{\gamma,d}$ for arbitrary $\gamma$. In particular, we prove
\begin{align}\label{eq:constants}
  S_{\gamma,d} \leq \frac{\pi}{\sqrt 3}\, L_{\gamma,d}^\cl &
  \qquad \text{if } 1\leq\gamma< 3/2 \text{ and } d \geq 1, \notag\\
  S_{\gamma,1} \leq 2\, L_{\gamma,1}^\cl &
  \qquad \text{if } 1/2\leq\gamma< 1 \text{ and }\, d = 1, \\
  S_{\gamma,d} \leq \frac{2\pi}{\sqrt 3}\, L_{\gamma,d}^\cl &
  \qquad \text{if } 1/2\leq\gamma< 1 \text{ and }\, d \geq 2, \notag
\end{align}
see Remark \ref{constants}. Moreover, the constants in the estimates for the number of negative eigenvalues satisfy
\begin{align}
  \label{eq:constantsnumber2}
    S_{0,2} \leq 6.04\, L_{0,2}^\cl &
  \qquad \text{if } \gamma=0 \text{ and } d = 2,\\
  \label{eq:constantsnumber3}
    S_{0,3} \leq 6.07\, L_{0,3}^\cl &
  \qquad \text{if } \gamma=0 \text{ and } d = 3,\\
  \label{eq:constantsnumber}
  S_{0,d} \leq 10.332\, L_{0,d}^\cl &
  \qquad \text{if } \gamma= 0 \text{ and }\, d \geq 4, 
\end{align}
see Remark \ref{constants} and Subsection \ref{sec:main1proof}. The upper bounds \eqref{eq:constantsnumber2}, \eqref{eq:constantsnumber3} can be supplemented by the lower bounds
\begin{align}
  \label{eq:constantsnumberlower2}
  S_{0,2} \geq 4 \, L_{0,2}^\cl &
  \qquad \text{if } \gamma=0 \text{ and } d = 2,\\
  \label{eq:constantsnumberlower3}
  S_{0,3} \geq 3 \, L_{0,3}^\cl &
  \qquad \text{if } \gamma=0 \text{ and } d = 3 ,
\end{align}
see Subsection \ref{sec:main1proof}. In particular, for $d=2$ the upper bound \eqref{eq:constantsnumber2} is off by at most a factor~$1.51$.


\subsection{Lifting with respect to dimension}\label{sec:main2proof}

In this subsection we use an argument in the spirit of Laptev and Weidl \cite{LW} to prove

\begin{theorem}\label{ltsharp}
  Let $d\geq 1$, $\gamma\geq 3/2$ and $\tau\geq 0$. Then
  \begin{equation}\label{eq:ltsharp}
    \tr\left[ H(v) +\tau \right]_-^\gamma
    \leq L_{\gamma,d}^\cl 
    \int_{\R^d} \left(v(x)_+^2-\tau\right)_+^{\gamma+d/2}\,dx
  \end{equation}
  with $L_{\gamma,d}^\cl$ defined in \eqref{eq:surfclass}.
\end{theorem}

Choosing $\tau=0$ and recalling \eqref{eq:weylsurf} we obtain Theorem \ref{main2}.

\begin{proof}
  We shall prove Theorem \ref{ltsharp} by induction over $d$. It is convenient to reflect the dependence on $d$ in the notation of the quadratic form, so we write
  \begin{equation*}
    h_d(v)[u] := \iint_{\R_+^{d+1}} |\nabla u|^2 \,dx\,dy 
    - \int_{\R^d} v(x) |u(x,0)|^2 \,dx\,dy  
  \end{equation*}
  and $H_d(v)$ for the associated operator. Note that this operator is also well-defined for $d=0$ and $v$ a non-negative real number. Indeed, in this case one has $H_0(v)u=-u''$ and $u'(0)=-vu(0)$ for $u\in\dom H_0(v)$, and one easily finds that $H_0(v)$ has one negative eigenvalue, namely $-v_+^2$. Hence
  \begin{equation*}
    \tr_{L_2(\R_+)}\left[ H_0(v) +\tau \right]_-^\gamma
    = \left(v_+^2-\tau\right)_+^{\gamma},
  \end{equation*}
  which is the analog of \eqref{eq:ltsharp} for $d=0$ and all $\gamma\geq 0$.
  
  Now we fix $d\geq 1$ and assume that the assertion is already proved for all smaller dimensions. We write $x=(x_1,x')$ with $x_1\in\R$, $x'\in\R^{d-1}$ and note that
  \begin{equation*}
    H_d(v)+\tau 
    \geq -\frac{d^2}{dx_1^2}\otimes 1_{L_2(\R^{d}_+)} 
    - \left[ H_{d-1}(v(x_1,\cdot))+\tau\right]_-
  \end{equation*}
  with the identification $L_2(\R^{d+1}_+) = L_2(\R)\otimes L_2(\R^{d}_+)$. Hence the variational principle and the operator-valued Lieb-Thirring inequality from \cite{LW} yields for all $\gamma\geq 3/2$
  \begin{equation}\label{eq:ltlw}
    \tr_{L_2(\R^{d+1}_+)}\left[ H_d(v) +\tau \right]_-^\gamma
    \leq L_{\gamma,1}^\cl \int_{-\infty}^\infty 
    \tr_{L_2(\R^{d}_+)}\left[ H_{d-1}(v(x_1,\cdot)) +\tau \right]_-^{\gamma+1/2}
    \,dx_1.
  \end{equation}
  By induction hypothesis, the right hand side is bounded from above by
  \begin{align*}
    L_{\gamma,1}^\cl L_{\gamma+1/2,d-1}^\cl
    \int_{-\infty}^\infty \left( \int_{\R^{d-1}} 
    (v(x_1,x')_+^2 - \tau)_+^{\gamma+d/2}
    \,dx'\right) \,dx_1
    = L_{\gamma,d}^\cl \int_{\R^{d}} (v(x)_+^2 - \tau)_+^{\gamma+d/2}\,dx\, ,
  \end{align*}
  which establishes the assertion for dimension $d$ and completes the proof of Theorem \ref{ltsharp}.
\end{proof}

\begin{remark}
  \label{constants}
  The above approach can be used to prove inequality \eqref{eq:main1} for $\gamma\geq 1/2$ and to obtain bounds \eqref{eq:constants} for the sharp constants $S_{\gamma,d}$. Indeed, according to \cite{HLW} and \cite{DLL} the operator-valued inequality \eqref{eq:ltlw} holds with an additional factor of $\pi/\sqrt 3$ on the right hand side if $\gamma\geq 1$, with an additional factor of $2$ if $1/2\leq \gamma<1$ and $d=1$ and with an additional factor of $2\pi/\sqrt 3$ if $\gamma\geq 1/2$ and $d\geq 2$.

Similarly, one can use the operator-valued inequality from \cite{FLS2} (see also \cite{H}) to prove \eqref{eq:main1} for $\gamma\geq 0$ and $d\geq 3$ and to obtain \eqref{eq:constantsnumber}. Extending the original proof of Lieb and Thirring \cite{LT} to the operator-valued case would yield \eqref{eq:main1} for $\gamma> 0$ and $d=2$. However, we do not know how to prove \eqref{eq:main1} with the operator-valued approach for $0<\gamma< 1/2$ if $d=1$ and for $\gamma=0$ if $d=2$. We give a proof based on a different idea in Subsection \ref{sec:main1proof} below.
\end{remark}


\subsection{Additional remarks}

\subsubsection{Magnetic fields}
Let $A\in L_{2,\loc}(\overline{\R_+^{d+1}},\R^{d+1})$ and let the operator $H(A,v)$ be defined through the closure of the quadratic form
\begin{equation*}
  \iint_{\R^{d+1}_+} |(-i\nabla-A) u|^2 \,dx\,dy - \int_{\R^d} v(x)|u(x,0)|^2\,dx,
  \quad u\in C_0^\infty(\overline{\R^{d+1}_+})\, .
\end{equation*}
By a similar argument as in \cite{LW} one can prove that Theorem \ref{ltsharp} remains true, with the same constant, if $H(v)$ is replaced by $H(A,v)$. More generally, all the inequalities sketched in Remark \ref{constants} remain true. The argument behind Theorem \ref{main1}, however, allows only $A$ which are independent of $y$ and orthogonal to the $y$-direction, see Remark \ref{dualityabstract}. To obtain the analog of \eqref{eq:main1} for the complete range of $\gamma$'s given in \eqref{eq:ltsurfass}, one can rely upon an abstract operator-theoretic argument, see \cite{R} for $\gamma=0$ and \cite{F} for $\gamma>0$.

\subsubsection{Leaky graph Hamiltonians}
In the previous subsections we studied eigenvalues of the Laplacian $H(v)$ on the halfspace with a perturbation by boundary conditions. This problem is essentially equivalent to the study of eigenvalues of the Schr\"odinger operator
$$
\tilde H(v) = -\Delta - v(x)\delta(y)
$$
in the whole space $\R^{d+1}$ with a potential supported on a hyperplane. The precise definition of the operator $\tilde H(v)$ in $L_2(\R^{d+1})$ is given via the quadratic form
$$
\iint_{\R^{d+1}} |\nabla u|^2 \,dx\,dy - \int_{\R^d} v(x)|u(x,0)|^2\,dx,
\quad u\in H^1(\R^{d+1})\, .
$$
The decomposition of a function into an even and an odd part with respect to the variable $y$ induces an orthogonal decomposition of the space $L_2(\R^{d+1})$, which reduces the operator $\tilde H(v)$. The part of $\tilde H(v)$ on odd functions is unitarily equivalent to the Dirichlet Laplacian on $\R_+^{d+1}$, whereas the part on even functions is unitarily equivalent to $H(\frac12 v)$. Hence
$$
\tr\left[\tilde H(v) \right]_-^\gamma 
= \tr\left[H(\tfrac 12 v) \right]_-^\gamma\,,
$$
and we obtain immediately the analogs of Theorems \ref{main1} and \ref{main2}.

\begin{theorem}
  \label{main1'}
  The inequality 
  \begin{equation}\label{eq:main1'}
    \tr\left[ \tilde H(v) \right]_-^\gamma
    \leq \tilde S_{\gamma,d} \int_{\R^d} v(x)_+^{2\gamma+d} \,dx
  \end{equation}
  holds for all $0\leq v\in L_{2\gamma+d}(\R^d)$ iff
  \begin{equation}
    \label{eq:ltsurfass'}
    \gamma>0 \quad\text{if } d=1,
    \qquad \text{and} \qquad
    \gamma\geq 0 \quad\text{if } d\geq 2 \, .
  \end{equation}
\end{theorem}

\begin{theorem}
  \label{main2'}
  Let $d\geq 1$ and $\gamma\geq 3/2$. Then the sharp constant in \eqref{eq:main1'} is $\tilde S_{\gamma,d} = 2^{-2\gamma-d} L_{\gamma,d}^\cl$.
\end{theorem}

\subsubsection{Complex-valued surface potentials}
In applications one often encounters the boundary value problem \eqref{eq:op}, \eqref{eq:bc} with a \emph{complex-valued} function $v$, which leads to non-real eigenvalues. If $v$ is sufficiently regular, the quadratic form \eqref{eq:form} generates an $m$-sectorial operator which we continue to denote by $H(v)$. We denote by $\lambda_j(v)$, $j=1,2,\ldots$, the (at most countably many) eigenvalues of $H(v)$ in the cut plane $\C\setminus [0,\infty)$, repeated according to their algebraic multiplicities. Following the approach suggested in \cite{FLLS} one obtains

\begin{theorem}
  Let $d\geq 1$ and $\gamma\geq 1$.
  \begin{enumerate}
    \item
      For eigenvalues with negative real part
      $$
      \sum_{j:\ \re\lambda_j(v)<0} (-\re\lambda_j(v))^\gamma
      \leq S_{\gamma,d} \int_{\R^d} (\re v(x))_-^{2\gamma+d} \,dx\, .
      $$
    \item
      If $\kappa>0$, then for eigenvalues outside the cone $\{|\im z|<\kappa\re z\}$
      $$
      \sum_{j:\ |\im\lambda_j(v)|\geq \kappa\re\lambda_j(v)} |\lambda_j(v)|^\gamma
      \leq 2^{1+\gamma+d/2} \left(1+\frac2\kappa\right)^{2\gamma+d} S_{\gamma,d}
      \int_{\R^d} |v(x)|^{2\gamma+d} \,dx\, .
      $$
  \end{enumerate}
  Here $S_{\gamma,d}$ is the constant from \eqref{eq:main1}.
\end{theorem}

\subsubsection{Waveguides}
Let $\omega\subset\R^d$ be a domain of finite measure and put $\Omega:=\omega\times\R_+$ and $\Gamma:=(\partial\omega)\times\R_+$. The quadratic form \eqref{eq:form}, restricted to $\{u\in H^1(\Omega): \ u=0 \text{ on } \Gamma \}$, defines a self-adjoint operator $H_\omega(v)$ in $L_2(\Omega)$, which corresponds to Dirichlet boundary conditions on $\Gamma$ and boundary conditions of the third type on $\omega\times\{0\}$. By the variational principle Theorem \ref{main1} implies that
\begin{equation*}
  \tr\left[ H_\omega(v) \right]_-^\gamma
  \leq S_{\gamma,d} \int_\omega v(x)_-^{2\gamma+d} \,dx
\end{equation*}
for $\gamma>0$ if $d=1$ and $\gamma\geq 0$ if $d\geq 2$ with the constant $S_{\gamma,d}$ from \eqref{eq:main1}. In particular, $S_{\gamma,d}=L_{\gamma,d}^\cl$ for $\gamma\geq 3/2$. We now show that in the special case where $v\equiv v_0$ is a constant, the estimate with the semi-classical constant holds already for $\gamma\geq 1$.

\begin{theorem}
  Let $\omega\subset\R^d$ be a domain of finite measure and $v\equiv v_0>0$ a constant. Then for any $\gamma\geq 1$
  \begin{equation}\label{eq:waveguide}
    \tr\left[ H_\omega(v) \right]_-^\gamma
    \leq L_{\gamma,d}^\cl\, |\omega|\, v_0^{2\gamma+d}\, .
  \end{equation}
\end{theorem}

\begin{proof}
  By separation of variables one has
  $$
  \tr\left[ H_\omega(v) \right]_-^\gamma 
  = \tr\left[-\Delta_\omega^D -v_0^2\right]_-^\gamma\,,$$
  where $-\Delta_\omega^D$ denotes the Dirichlet Laplacian on $\omega$. Therefore the assertion follows from the Berezin-Li-Yau inequality; see \cite{B,LY} and also \cite{La}.
\end{proof}

The same argument shows that if $\omega$ is tiling (in particular, \emph{any} interval $\omega$ if $d=1$), then \eqref{eq:waveguide} holds for all $\gamma\geq 0$; see \cite{P}.


\section{Relativistic Schr\"odinger operators}

\subsection{Statement of the results}
In this section we derive a connection between Schr\"o\-din\-ger operators $H(v)$ with surface potential in $L_2(\R_+^{d+1})$ and relativistic Schr\"odin\-ger operators $\sqrt{-\Delta}-v$ in $L_2(\R^{d})$. We begin by recalling Lieb-Thirring and Cwikel-Lieb-Rozenblum inequalities for the latter operator.

\begin{proposition}\label{daubechies}
  The inequality 
  \begin{equation}\label{eq:daubechies}
    \tr\left[ \sqrt{-\Delta}-v \right]_-^\gamma
    \leq D_{\gamma,d} \int_{\R^d} v(x)_+^{\gamma+d} \,dx
  \end{equation}
  holds for all $0\leq v\in L_{\gamma+d}(\R^d)$ iff
  \begin{equation}
    \label{eq:daubechiesass}
    \gamma>0 \quad\text{if } d=1,
    \qquad \text{and} \qquad
    \gamma\geq 0 \quad\text{if } d\geq 2 \, .
  \end{equation}
\end{proposition}

This result is due to Daubechies \cite{D}. The fact that the inequality is not valid for $\gamma=0$ if $d=1$ follows from the fact that $\sqrt{-\Delta}-v$ has a negative eigenvalue for any non-trivial $v\geq 0$ if $d=1$. This can be proved as in \cite[Prop. 7.4]{S}.

The Weyl-type asymptotics in the relativistic case read
\begin{equation}\label{eq:weylrel}
  \lim_{\alpha\to\infty}
  \alpha^{-\gamma-d}\tr\left[ \sqrt{-\Delta} - \alpha\, v \right]_-^\gamma
  = D_{\gamma,d}^\cl \int_{\R^d} v(x)_+^{\gamma+d}\,dx\, ,
\end{equation}
with
\begin{equation}
  \label{eq:relclass}
  D_{\gamma,d}^\cl 
  := 2^{-d} \pi^{-d/2} \frac{\Gamma(\gamma+1)\,\Gamma(d+1)}{\Gamma(\gamma+d+1)\,\Gamma(d/2+1)}\, .
\end{equation}
We denote by $D_{\gamma,d}$ the \emph{sharp} constant in \eqref{eq:daubechies}. In the case $d=3$, the bound
\begin{equation}
  \label{eq:daubechiesconst3}
  D_{\gamma,3}\leq 6.08\, D_{\gamma,3}^\cl\, ,
  \quad\gamma\geq 0\, ,
\end{equation}
is contained in \cite{D}. Similarly one proves that for $d=2$
\begin{equation}
  \label{eq:daubechiesconst2}
  D_{\gamma,2}\leq 6.04\, D_{\gamma,2}^\cl\, ,
  \quad\gamma\geq 0\, .
\end{equation}

We are now in position to state a result which connects relativistic Schr\"odinger operator and non-relativistic Schr\"odinger operators with surface potentials. As usual, we denote by $N(-\tau,T)$ the number of eigenvalues, counting multiplicities, less than $-\tau$ of a self-adjoint, lower semi-bounded operator $T$, and write $N(T):=N(0,T)$.

\begin{theorem}\label{duality}
  Assume that $v$ is form-compact with respect to $\sqrt{-\Delta}$. Then for any $\tau\geq 0$ one has 
  \begin{equation}
    \label{eq:dualitynumber}
    N(-\tau,H(v))=N(\sqrt{-\Delta+\tau}-v)\, .
  \end{equation}
  Moreover, for any $\gamma>0$ and $0<\rho<1$ one has
  \begin{equation}\label{eq:dualitymoments}
    \tr\left[\sqrt{-\Delta}-v\right]_-^\gamma
    \leq \tr\left[H(v)\right]_-^{\gamma/2}
    \leq \left(\frac\rho{\sqrt{1-\rho^2}}\right)^\gamma 
    \tr\left[\sqrt{-\Delta}-\rho^{-1} v\right]_-^\gamma\, .
  \end{equation}
\end{theorem}

We shall prove this in Subsection \ref{sec:duality} below, as well as the following

\begin{corollary}\label{dualityconstants}
  The sharp constants in \eqref{eq:main1} and \eqref{eq:daubechies} coincide for $\gamma=0$ and $d\geq 2$,
  \begin{equation}\label{eq:dualitynumber2}
  S_{0,d} = D_{0,d} \, ,
  \end{equation}
  and satisfy for any $\gamma>0$ and $d\geq 1$
  \begin{equation}\label{eq:dualitymoments2}
  \frac{\gamma^{\gamma/2}\, d^{d/2}}{(\gamma+d)^{(\gamma+d)/2}}\, 
  S_{\gamma/2,d}
  \leq D_{\gamma,d} \leq S_{\gamma/2,d}\, .
  \end{equation}
\end{corollary}

We shall use Theorem \ref{duality} in two directions. In Subsection \ref{sec:main1proof} we shall use the known Lieb-Thirring inequalities in the relativistic case to derive the Lieb-Thirring inequalities for surface potentials. In Subsection \ref{sec:bks} we shall use the estimates on the constants $S_{\gamma,d}$ for surface potentials to improve upon the estimates \eqref{eq:daubechiesconst3} and \eqref{eq:daubechiesconst2} in the relativistic case. We also discuss the connection of our inequality with an inequality by Birman, Koplienko and Solomyak.


\subsection{Duality}\label{sec:duality}

The following lemma characterizes the negative eigenvalues of the operator $H(v)$ as the values $-\tau$ for which $0$ is an eigenvalue of the operator $\sqrt{-\Delta+\tau}-v$.

\begin{lemma}\label{dualitylemma}
  Assume that $v$ is form-compact with respect to $\sqrt{-\Delta}$ and let $\tau> 0$.
  \begin{enumerate}
    \item
      Let $f\in\ker(\sqrt{-\Delta+\tau}-v)$ and define $u(x,y):=(\exp(-y\sqrt{-\Delta+\tau})f)(x)$. Then $u\in\ker(H(v)+\tau)$ and $u(x,0)=f(x)$.
    \item
      Let $u\in\ker(H(v)+\tau)$ and define $f(x):=u(x,0)$. Then $f\in\ker(\sqrt{-\Delta+\tau}-v)$ and $u(x,y)=(\exp(-y\sqrt{-\Delta+\tau})f)(x)$.
  \end{enumerate}
\end{lemma}

The proof of this lemma is straightforward and will be omitted (see \cite{FS} for a similar argument). Using a modification of the Birman-Schwinger principle we now give the

\begin{proof}[Proof of Theorem \ref{duality}]
  Since the eigenvalues of the operators $\sqrt{-\Delta+t}-v$ are increasing with respect to $t$, one has for any fixed $\tau\geq 0$
  $$
  N(\sqrt{-\Delta+\tau}-v) = \#_\textrm{m} \{ t>\tau:\, 0 \text{ is an eigenvalue of } 
\sqrt{-\Delta+t}-v \} \, .
  $$
  Here $\#_\textrm{m}\{\ldots\}$ means that the cardinality of $\{\ldots\}$ is determined according to multiplicities. By Lemma \ref{dualitylemma}, the right hand side coincides with
  $$
  \#_\textrm{m} \{ t>\tau:\, -t \text{ is an eigenvalue of } H(v) \}
  = N(-\tau,H(v))\, ,
  $$
  as claimed.

  To prove \eqref{eq:dualitymoments} we note that by the previous argument
  \begin{equation}\label{eq:dualityint}
    \tr\left[H(v)\right]_-^{\gamma/2} 
    = \frac\gamma2 \int_0^\infty N(-\tau, H(v))\, \tau^{\gamma/2 -1}\,d\tau
    = \frac\gamma2 \int_0^\infty N(\sqrt{-\Delta+\tau}-v) \,
    \tau^{\gamma/2 -1}\,d\tau\, .
  \end{equation}
  The elementary inequalities
  \begin{equation*}
    \rho\sqrt\lambda +\sqrt{1-\rho^2}\sqrt\tau
    \leq\sqrt{\lambda+\tau}
    \leq\sqrt\lambda + \sqrt\tau,
    \quad \lambda,\tau>0,\, 0<\rho<1\, ,
  \end{equation*}
  imply
  \begin{equation*}
    N(-\sqrt\tau,\sqrt{-\Delta} -v)
    \leq N(\sqrt{-\Delta+\tau}-v) 
    \leq N(-\sqrt{1-\rho^2}\sqrt\tau,\rho\sqrt{-\Delta} -v)\, .
  \end{equation*}
  Plugging this into \eqref{eq:dualityint} we obtain \eqref{eq:dualitymoments}.
\end{proof}

\begin{proof}[Proof of Corollary \ref{dualityconstants}]
  Equality \eqref{eq:dualitynumber2} as well as the second inequality in \eqref{eq:dualitymoments2} follow immediately from equality \eqref{eq:dualitynumber} and the first inequality in \eqref{eq:dualitymoments}. To prove the first inequality in \eqref{eq:dualitymoments2} we combine Daubechies' inequality \eqref{eq:daubechies} with the second inequality in \eqref{eq:dualitymoments} to get
  $$
  \tr\left[H(v)\right]_-^{\gamma/2}
  \leq D_{\gamma,d} \left(\frac\rho{\sqrt{1-\rho^2}}\right)^\gamma \rho^{-\gamma-d} \int_{\R^d} v(x)_+^{\gamma+d}\,dx \, .
  $$
  The assertion follows by optimizing over $0<\rho<1$.
\end{proof}

\begin{remark}\label{dualityabstract}
	The material in this subsection, except for the proof of the second part of Corollary \ref{dualityconstants}, is of abstract nature. If $A$ is a non-negative operator in a Hilbert space $\mathfrak H$ and $B$ is a self-adjoint operator which is relatively form-compact with respect to $A$, define the operator $H$ in $L_2(\R_+,\mathfrak H)$ by the quadratic form
	$$
		\int_0^\infty \left( \|F'(y)\|_\mathfrak H^2 + \|A F(y)\|_\mathfrak H^2 \right)\,dy - b[F(0)]
	$$
	for $F\in H^1(\R_+,\mathfrak H)\cap L_2(\R_+,\dom A)$. Here $b$ is the quadratic form of $B$. Then the argument of this subsection yields
	$$
		N(-\tau,H)=N(\sqrt{A^2+\tau}-B)\,.
	$$
	As an application of this generalization one can extend Theorem \ref{duality} to relativistic Schr\"odin\-ger operators with magnetic field or to relativistic Schr\"odinger operators with a Hardy weight subtracted (see \cite{FLS1}).
\end{remark}


\subsection{Proof of Theorem \ref{main1}}\label{sec:main1proof}

Theorem \ref{main1} is an immediate consequence of Proposition~\ref{daubechies} and Theorem \ref{duality}. 

Moreover, Daubechies' bounds \eqref{eq:daubechiesconst2}, \eqref{eq:daubechiesconst3} yield the upper bounds \eqref{eq:constantsnumber2}, \eqref{eq:constantsnumber3} for the sharp constants $S_{0,d}$. Similarly, the lower bounds \eqref{eq:constantsnumberlower2}, \eqref{eq:constantsnumberlower3} follow from
\begin{equation}
\label{eq:daubechieslower}
D_{0,d} 
\geq  \frac{2^{d-1}}{(d-1)^d}\, \Gamma\left(d +1\right) D_{0,d}^\cl\, ,
\qquad d\geq 2.
\end{equation}
(Note that this is only useful for $d\leq 7$, since otherwise the factor on the right hand side is smaller than one and the bound $D_{0,d} \geq D_{0,d}^\cl$ follows from \eqref{eq:weylrel}.) The lower bound \eqref{eq:daubechieslower} can be seen as follows. The definition of $D_{0,d}$ implies that if $\int v^d\,dx < D_{0,d}^{-1}$, then $\sqrt{-\Delta}-v$ is a non-negative operator. Hence
$$
\int_{\R^d} v |u|^2\,dx \leq \int_{\R^d} |(-\Delta)^{1/4}u|^2 \,dx
$$
for all $u\in H^{1/2}(\R^d)$. Choosing $v=\alpha |u|^{2/(d-1)}$ with $\alpha$ such that $\alpha^d \int |u|^{2d/(d-1)}\,dx = (D_{0,d}+\epsilon)^{-1}$ and letting $\epsilon$ tend to zero, we find
\begin{equation*}
D_{0,d}^{-1/d} \left(\int_{\R^d} |u|^{2d/(d-1)}\,dx \right)^{(d-1)/d}
\leq \int_{\R^d} |(-\Delta)^{1/4} u|^2\,dx
\end{equation*}
for all $u\in H^{1/2}(\R^d)$. Hence $D_{0,d}^{-1/d}$ is not larger than the constant in the sharp Sobolev inequality
\begin{equation}\label{eq:sobolevrel}
S'_d\,  \|u\|_{2d/(d-1)}^2 \leq \left\|(-\Delta)^{1/4} u\right\|^2,
\qquad
S'_d := \frac{d-1}2\, 2^{1/d}\, \pi^{(d+1)/2d}\, \Gamma\left(\frac{d+1}2\right)^{-1/d}.
\end{equation}
see \cite[Thm. 8.4]{LL}. Recalling definition \eqref{eq:relclass} of $D_{\gamma,d}^\cl$ we arrive at \eqref{eq:daubechieslower}.

\begin{remark}
	Instead of using the `relativistic' Sobolev inequality \eqref{eq:sobolevrel} to prove \eqref{eq:daubechieslower} we could have used a similar argument based on the sharp Sobolev trace inequality
	\begin{equation}\label{eq:sobolevtrace}
		S'_d\left( \int_{\R^d} |u|^{2d/(d-1)}\,dx\right)^{(d-1)/d} \leq \iint_{\R^{d+1}_+} |\nabla u|^2\,dx\,dy\,,
	\end{equation}
	to directly prove \eqref{eq:constantsnumberlower2}, \eqref{eq:constantsnumberlower3}. The constant $S'_d$ in \eqref{eq:sobolevtrace} is the same as in \eqref{eq:sobolevrel}, see \cite{E}. Indeed, an argument similar to our Lemma \ref{dualitylemma} was used in \cite{CL} to derive \eqref{eq:sobolevtrace} from \eqref{eq:sobolevrel}.
\end{remark}

\begin{remark}
  \label{weyl}
  Weyl-type asymptotics \eqref{eq:weylsurf} can be proved by a bracketing argument, dividing $\R^{d+1}_+$ into domains $Q\times\R_+$ with $Q\subset\R^d$ a small cube. An alternative proof can be based on Theorem \ref{duality}. Indeed, for $\gamma=0$ the asymptotics \eqref{eq:weylsurf} follow immediately from Theorem~\ref{duality} and \eqref{eq:weylrel} (which is valid for all smooth $v$). If $\gamma>0$ we write as in \eqref{eq:dualityint}
  $$
  	\tr\left[ H(\alpha\, v) \right]_-^\gamma 
  	= \gamma \int_0^\infty N(-\tau, H(\alpha\, v))\, \tau^{\gamma-1}\,d\tau
  	= \gamma \int_0^\infty N(\sqrt{-\Delta+\tau}- \alpha\, v)\, \tau^{\gamma-1}\,d\tau\, .
  $$
  For smooth $v$ one can justify that this is asymptotically equal as $\alpha\to\infty$ to
  \begin{align*}
  	& \gamma \int_0^\infty \iint_{\{(x,\xi)\in \R^d\times\R^d :\ \sqrt{|\xi|^2+\tau}- \alpha\, v(x)<0 \}} \frac{dx\,d\xi}{(2\pi)^d} \, \tau^{\gamma-1}\,d\tau \\
  	& \qquad = \gamma L_{0,d}^\cl \int_0^\infty \int_{\R^d} \left((\alpha v(x))_+^2 - \tau\right)_+^{d/2} \,dx \, \tau^{\gamma-1}\,d\tau \\
  	& \qquad = \alpha^{2\gamma+d} L_{\gamma,d}^\cl \int_{\R^d} v(x)_+^{2\gamma+d} \,dx \,.
	\end{align*}
	This concludes the sketch of \eqref{eq:weylsurf}. We note that by a standard argument based on Theorem~\ref{main1}, the asymptotics \eqref{eq:weylsurf} extend to all $v$ for which the right hand side is finite if $\gamma>0$ and $d=1$ or if $\gamma\geq 0$ and $d\geq 2$.
\end{remark}


\subsection{Relation with the BKS inequality}\label{sec:bks}

In this subsection we shall use Theorem \ref{duality} to improve upon known constants for relativistic Schr\"odinger operators and discuss its relation with an inequality by Birman, Koplienko and Solomyak. We begin with a result about massive relativistic Schr\"odinger operators.

\begin{remark}\label{massive}
Let $d\geq 3$ and $m\geq 0$. Then
\begin{equation}\label{eq:massive}
N(\sqrt{-\Delta+m^2}-m -v) \leq 10.332\, D_{0,d}^{\cl} \int_{\R^d} \left((v(x)+m)_+^2-m^2\right)_+^{d/2}\,dx\,.
\end{equation}
This improves upon Daubechies' bound \cite{D} who obtains \eqref{eq:massive} with constant $14.14\, D_{0,3}^{\cl}$ for $d=3$. To prove \eqref{eq:massive} we combine Theorem \ref{duality} and Remark \ref{ltsharp} to get
$$
N(\sqrt{-\Delta+m^2}-m -v) = N(-m^2,H(v+m)) 
\leq 10.332\, L_{0,d}^\cl \int_{\R^d} \left((v(x)+m)_+^2-m^2\right)_+^{d/2}\,dx\,,
$$
and recall that $L_{0,d}^\cl=D_{0,d}^{\cl}\,$.
\end{remark}

We return again to the massless case $m=0$.

\begin{remark}\label{daubechiesconstimproved}
The sharp constants in \eqref{eq:daubechies} satisfy
\begin{align}
  \label{eq:daubechiesconstimproved}
  D_{\gamma,2} \leq \sqrt 3 \pi\, D_{\gamma,2}^\cl &
  \qquad \text{if } 2\leq\gamma<3 \text{ and } d = 2,\notag \\
  D_{\gamma,2} \leq 4 \, D_{\gamma,2}^\cl &
  \qquad \text{if } \gamma\geq 3 \text{ and } d = 2,\\
  D_{\gamma,3} \leq \frac{15\pi}8 \, D_{\gamma,3}^\cl &
  \qquad \text{if } \gamma\geq 3 \text{ and }\, d =3. \notag
\end{align}
This improves upon \eqref{eq:daubechiesconst2} and \eqref{eq:daubechiesconst3}. To prove the first inequality in \eqref{eq:daubechiesconstimproved} we combine \eqref{eq:dualitymoments2} with \eqref{eq:constants} to get
$$
D_{2,2} \leq \frac\pi{\sqrt3} \, L_{1,2}^\cl = \sqrt 3 \pi\, D_{2,2}^\cl\, .
$$
By the argument of Aizenman and Lieb \cite{AL}, this implies $D_{\gamma,2} \leq \sqrt 3 \pi\, D_{\gamma,2}^\cl$ for all $\gamma\geq 2$. The other bounds in \eqref{eq:daubechiesconstimproved} are proved similarly.
\end{remark}

In conclusion we would like to recall a result by Birman, Koplienko and Solomyak.

\begin{proposition}\label{bks}
  Let $0<s<1$, $\gamma\geq1$ and $A$, $B$ non-negative operators such that $\tr(A-B)_+^{s\gamma}<\infty$. Then
  \begin{equation}
    \label{eq:bks}
    \tr(A^s-B^s)_+^\gamma \leq \tr(A-B)_+^{s\gamma}
  \end{equation}
\end{proposition}

In \cite{BKS} this is proved under the additional assumption $A\geq B$, but, as observed in \cite{LSS}, this assumption can be removed in view of the operator inequality
$$
A^s-B^s \leq \left( B + (A-B)_+ \right)^s -B^s \, .
$$
Moreover, \cite{LSS} contains an elementary proof of \eqref{eq:bks} in the case $\gamma=1$. We deduce from \eqref{eq:bks} and \eqref{eq:lt} that
\begin{equation}\label{eq:bksschroedinger}
	\tr(\sqrt{-\Delta}-v)_-^\gamma \leq \tr(-\Delta-v_+^2)_-^{\gamma/2} \leq L_{\gamma/2,d} \int_{\R^d} v(x)_+^{\gamma+d}\,dx,
	\quad \gamma\geq 1.
\end{equation}
Since the best known bounds on the constants $L_{\gamma/2,d}$ for $\gamma\geq 1$ coincide with those for $S_{\gamma/2,d}$, Proposition \ref{bks} yields for $\gamma\geq 1$ the same bounds on $D_{\gamma,d}$ as our Theorem \ref{duality}. In particular, Remark \ref{daubechiesconstimproved} can also be derived via \eqref{eq:bksschroedinger}. In contrast, the $\gamma=0$ result of Remark \ref{massive} cannot be deduced via \eqref{eq:bksschroedinger}. It is interesting, in our opinion, to understand whether there is a deeper connection between the $s=1/2$ case of Proposition \ref{bks} and Theorem \ref{duality}.


\subsection*{Acknowledgments}  
This work has been supported by DAAD grant D/06/49117 (R. F.).


\bibliographystyle{amsalpha}

\begin{thebibliography}{FLLS}

\bibitem[AL]{AL} M. Aizenman, E. Lieb, \textit{On semiclassical
  bounds for eigenvalues of Schr\"odinger operators}. Phys. Lett. A
  \textbf{66} (1978), no. 6, 427--429.
\bibitem[B]{B} F.A.~Berezin, \textit{Covariant and contravariant
  symbols of operators} [Russian]. Math. USSR Izv. {\bf 6} (1972),
  1117--1151.
\bibitem[BKS]{BKS}
  M. S. Birman, L. S. Koplienko, M. Z. Solomyak, \textit{Estimates for the spectrum of the difference between fraction powers of two self-adjoint operators}. Soviet Math. (Iz. VUZ) \textbf{19} (1975), no. 3, 1--6.
\bibitem[CL]{CL}
	E. Carlen, M. Loss, \textit{Competing symmetries of some functionals arising in mathematical physics}. In: Stochastic processes, physics and geometry (Ascona and Locarno, 1988), 277--288, World Sci. Publ., Teaneck, NJ, 1990.
\bibitem[CS]{CS}
	L. Caffarelli, L. Silvestre, \textit{An extension problem related to the fractional Laplacian}.
Comm. Partial Differential Equations \textbf{32} (2007), no. 8, 1245--1260.
\bibitem[D]{D}
  I. Daubechies, \textit{An uncertainty principle for fermions with generalized kinetic energy}.  Comm. Math. Phys. \textbf{90} (1983), no. 4, 511--520.
\bibitem[DLL]{DLL}
  J.~Dolbeault, A.~Laptev, M.~Loss, \textit{Lieb-Thirring inequalities
    with improved constants}. J. Eur. Math. Soc., to appear.
\bibitem[E]{E}
	J. Escobar, \textit{Sharp constant in a Sobolev trace inequality}. Indiana Univ. Math. J. \textbf{37} (1988), 687–-698.
\bibitem[F]{F}
	R. L. Frank, \textit{A remark on eigenvalue estimates and semigroup domination}. In preparation.
\bibitem[FLLS]{FLLS}
  R. L. Frank, A. Laptev, E. H. Lieb, R. Seiringer, \textit{Lieb-Thirring inequalities for Schr\"odinger operators with complex-valued potentials}. Lett. Math. Phys. \textbf{77} (2006), 309--316.
\bibitem[FLS1]{FLS1}
  R. L. Frank, E. H. Lieb, R. Seiringer, \textit{Hardy-Lieb-Thirring inequalities for fractional Schr\"odinger operators}. J. Amer. Math. Soc., to appear.
\bibitem[FLS2]{FLS2}
  R. L. Frank, E. H. Lieb, R. Seiringer, \textit{Number of bound states of Schr\"odinger operators with matrix-valued potentials}. Lett. Math. Phys., to appear.
\bibitem[FS]{FS}
  R. L. Frank, R. G. Shterenberg, \textit{On the scattering theory of the Laplacian with a periodic boundary condition. II. Additional channels of scattering}.  Doc. Math. \textbf{9} (2004), 57--77.
\bibitem[H1]{H}
    D.~Hundertmark, \textit{On the number of bound states for Schr\"odinger
operators with operator-valued potentials}. Ark. Mat. \textbf{40}
(2002), 73--87.
\bibitem[H2]{H2}
    D.~Hundertmark, \textit{Some bound state problems in quantum mechanics}. In: Spectral theory and mathematical physics: a Festschrift in honor of Barry Simon's 60th birthday, 463--496, Proc. Sympos. Pure Math. \textbf{76}, Part 1, Amer. Math. Soc., Providence, RI, 2007.
\bibitem[HLW]{HLW}
D. Hundertmark, A. Laptev and T. Weidl, \textit{New bounds on the
Lieb-Thirring constants}. Invent. Math., \textbf{40} (2000), 693--704.
\bibitem[L]{La} A. Laptev, \textit{Dirichlet and Neumann eigenvalue
    problems on domains in Euclidean
    spaces}. J. Funct. Anal. \textbf{151} (1997), no. 2, 531--545.
\bibitem[LW1]{LW} 
  A. Laptev, T. Weidl, \textit{Sharp Lieb-Thirring inequalities in
  high dimensions}. Acta Math. \textbf{184} (2000), no. 1, 87--111.
\bibitem[LW2]{LW2} 
  A. Laptev, T. Weidl, \textit{Recent results on Lieb-Thirring
  inequalities}. Journ\'ees "\'Equations aux D\'eriv\'ees Partielles"
  (La Chapelle sur Erdre, 2000), Exp. No. XX, Univ. Nantes, Nantes,
  2000.
\bibitem[LY]{LY} P.~Li, S-T.~Yau, \textit{On the Schr\"odinger equation and
    the eigenvalue problem}. Comm. Math. Phys. {\bf 88} (1983),
    309--318.
\bibitem[LL]{LL} E. H. Lieb, M. Loss, \textit{Analysis. Second
     edition}. Graduate Studies in Mathematics \textbf{14}, American
     Mathematical Society, Providence, RI, 2001.
\bibitem[LSS]{LSS} E. H. Lieb, H. Siedentop, J. P. Solovej, \textit{Relativistic electrons in classical electromagnetic fields}. J.~Stat. Phys. \textbf{89} (1997), 37--59.
\bibitem[LT]{LT} E.~H.~Lieb, W. Thirring, \textit{Inequalities for the
  moments of the eigenvalues of the Schr\"odinger Hamiltonian and
  their relation to Sobolev inequalities}. Studies in Mathematical
  Physics, 269--303. Princeton University Press, Princeton, NJ, 1976.
\bibitem[P]{P} G.~P\'olya, \textit{On the eigenvalues of vibrating
  membranes}. Proc. London Math. Soc. {\bf 11} (1961), 419--433.
\bibitem[R]{R}
G.~V.~Rozenblyum, \textit{Domination of semigroups and estimates for eigenvalues}. St. Petersburg Math. J. \textbf{12} (2001), no. 5, 831--845.
\bibitem[S]{S} B. Simon, {\it Trace ideals and their
      applications}, Second edition, Mathematical Surveys and Monographs
    {\bf 120}, American Mathematical Society, Providence, RI, 2005.

\end{thebibliography}

\end{document}